\documentclass[conference]{IEEEtran}
\usepackage{amsthm}
\newtheorem{theorem}{Theorem}[section]

\newtheorem{lemma}[theorem]{Lemma}
\usepackage{color}

\usepackage{cite}

\usepackage{graphicx}
\ifCLASSINFOpdf
\else
\fi
\usepackage{amsmath}

\usepackage{stfloats}
\usepackage{url}

\begin{document}
%
\title{A Graph-theoretic Model to Steganography on Social Networks}

\author{\IEEEauthorblockN{Hanzhou Wu$^1$, Wei Wang$^1$, Jing Dong$^1$, Yiliang Xiong$^2$ and Hongxia Wang$^3$}
\IEEEauthorblockA{$^1$Institute of Automation, Chinese Academy of Sciences, Beijing 100190, China\\
$^2$Microsoft (China) Co., Ltd., Suzhou 215123, China\\
$^3$School of Inf. Science and Technology, Southwest Jiaotong University, Chengdu 611756, China\\
Email: h.wu.phd@ieee.org}}


%


\maketitle

\begin{abstract}
Steganography aims to conceal the very fact that the communication takes place, by embedding a message into a digit object such as image without introducing noticeable artifacts. A number of steganographic systems have been developed in past years, most of which, however, are confined to the laboratory conditions where the real-world use of steganography are rarely concerned. In this paper, we introduce an alternative perspective to steganography. A graph-theoretic model to steganography on social networks is presented to analyze real-world steganographic scenarios. In the graph, steganographic participants are corresponding to the vertices with meaningless unique identifiers. Each edge allows the two vertices to communicate with each other by any steganographic algorithm. Meanwhile, the edges are associated with weights to quantize the corresponding communication risk (or say cost). The optimization task is to minimize the overall risk, which is modeled as additive over the social network. We analyze different scenarios on a social network, and provide the suited solutions to the corresponding optimization tasks. We prove that a multiplicative probabilistic graph is equivalent to an additive weighted graph. From the viewpoint of an attacker, he may hope to detect suspicious communication channels, the data encoder(s) and the data decoder(s). We present limited detection analysis to steganographic communication on a network.
\end{abstract}


%
\IEEEpeerreviewmaketitle

\section{Introduction}
Any communication system that aims to convey a message to the intentional receiver while conceal the very fact that the communication takes place can be classified as steganography \cite{kodovsky:thesis}. A most important yet still challenging requirement for a secure steganographic system is that it should be impossible for an eavesdropper to distinguish between ordinary objects and objects that contain hidden information \cite{jessica:book}. Basically, both steganography and cryptography provide secret communication. Cryptography, however, often exposes clear marks on the encrypted data for an eavesdropper to trace down, whereas steganography even conceals the presence of communication.

Steganography works by hiding a message within a digital object \cite{cox:book}. The resulting stego object containing the hidden information has no noticeable perceptual difference to the original host. It will be sent to the receiver via some insecure public channel such as the Internet. The goal of the decoder is to reconstruct the hidden information for subsequent purpose. Therefore, a secure steganographic system always requires us to develop such embedding and extraction algorithms that the stego object should not cause any suspicion \cite{jessica:book} and the hidden information can be reliably recovered at the decoder side.

A number of advanced steganographic systems such as \cite{holub:wow, holub:suniward, guo:uedr, li:tifs, tang:gan} have been developed to achieve the secret communication. The existing works have moved the underlying research of steganography ahead. However, with the rapidly development of information technologies, especially for social medias such as Facebook, Twitter and Instagram etc., there still has a long way to large-scale practice for steganography. A reasonable explanation for why steganography has not become common use like cryptography in social network services may be the lack of an adequate and real-world benchmark, which leads us to often perform the steganographic experiments on synthetic data through a local (or offline) manner. It would be quite desirable to study the social behaviors, protocols and any other possible scenarios of steganography. In this sense, we may not care about the details of the used steganographic algorithms, but rather quantize the concerned characteristics of steganography as analyzable scalars for possible optimization, which may provide us the access to move steganography from laboratory into real-world.

Wikipedia\footnote{Online available: \url{https://en.wikipedia.org/wiki/Social_network}} defines the social network as a social structure consisting of a set of social actors such as individuals and organizations, sets of dyadic ties, and social interactions between actors. The social network perspective provides lots of methods for analyzing the structure of entire social entities and theories explaining the patterns observed in these structures. The study of structures uses social network analysis, especially graph theory, to identify local or global patterns, locate influential entities, examine network dynamics, and address many optimization tasks. Steganography is essentially a communication task. It is straightforward to model steganographic activities on a social network. In the network, the vertices correspond to the people such as the data encoders and decoders or other social entities such as information routing devices. The edges represent the social links \cite{www} between them, which allow the vertices to communicate with each other by any steganographic means.

We present a simple graph-theoretic model to steganography on social networks in this paper. We will not pay attention to the design of a steganographic scheme that relies on some specified cover, but rather model it as a general communication flow on a network. The goal is to reliably convey messages from the encoders to the decoders with the lowest communication risk (or say cost) via the edges corresponding to the insecure channels. The edges are associated with weights that reflect the corresponding communication risks. A successful steganographic communication between any two vertices will correspond to a network path with a low risk. We model the overall risk as additive over the social network, which enables us to analyze different steganography-based communication scenarios and deal with different optimization tasks. From the viewpoint of an attacker, he may detect suspicious communication channels and even the data encoder(s) and the data decoder(s). We therefore further analyze the detection for steganographic network.

The rest of this paper are organized as follows. We formulate our problem in Section II. In Section III, We analyze different steganographic communication scenarios on social networks and provide reliable solutions. A probabilistic perspective has also been investigated. In Section IV, we present approaches to detecting steganographic communication. Finally, we conclude this paper in Section V.

\section{Problem Formulation}
A social network corresponds to a graph structure $G(V, E)$, where $V = \{v_1, v_2, ..., v_n\}$ denotes the set of vertices and $E = \{e_1, e_2, ..., e_m\}$ represents the set of edges. Here, every edge $e_k\in E$ corresponds to a vertex-pair, namely $e_k = (v_i,v_j)$ for some $1\le i\neq j\le n$. An undirected graph is a graph in which edges have no orientation. This indicates that, $(v_i,v_j)$ is equivalent to $(v_j,v_i)$. In this paper, we are to study undirected graph. A path (if any) between two vertices $v_i$ and $v_j$ is corresponding to such a vertex sequence $(v_{q_1}, v_{q_2}, ..., v_{q_t})$ that $v_{q_1} = v_i$, $v_{q_t} = v_j$ and $(v_{q_{k-1}}, v_{q_k})\in E$ for all $2\leq k\leq t$.

We model steganography on a social network $G(V, E)$. The vertices correspond to steganographic participants such as the data encoders, decoders, even the potential attackers, and other social entities such as servers, information routing devices. The edges show the steganographic communication links between vertices. It allows a message to be conveyed along the edges. Each $e_i \in E$ ($1\leq i\leq m$) is associated with a weight $w_i> 0$ to reveal the steganographic communication risk, which may involve the possibility of being attacked, the cost of bandwidth resource and other possible expenses. A path (if any) between $v_i$ and $v_j$ ($1\leq i\neq j\leq n$) implies that, $v_i$ can share a message with $v_j$ along the path with the risk that depends on the assigned weights of the edges belonging to the path.

Mathematically, we use $S = \{s_1, s_2,$ $..., s_{n_1}\}$ $\subset V$ and $T = \{t_1, t_2,$ $..., t_{n_2}\}$ $\subset V$ to denote the encoder set and decoder set. For each $s_i \in S$, we define its individual decoder set as $T_i\subset T$, meaning that, $s_i$ hopes to send a message to each of $T_i$. Here, $\cup_{i=1}^{n_1}T_i = T$. Let $P(v_i, v_j)$ be a set including all paths between $v_i$ and $v_j$. Each path in $P(v_i, v_j)$ is corresponding to a subset of $E$. Let $P(v_i, v_j)$ = $\{L_1(v_i, v_j),$ $L_2(v_i, v_j),$ ..., $L_{|P(v_i, v_j)|}(v_i, v_j)\}$. Here, $|*|$ means the size of a set. For compactness, we will sometimes consider $L_k(v_i, v_j)\in P(v_i, v_j)$ as a set including all involved edges, i.e., $L_k(v_i, v_j)\subset E$.

We assume that $G(V, E)$ is connected, namely, there always exists at least one path between any two vertices. Otherwise, we should separately analyze all connected subgraphs since two vertices belonging to two different connected components will be never able to send messages to each other. The optimization task is to find such a subset of $E$ that it enables all $s_i\in S$ to send a message to all decoders in $T_i$, and the overall risk can be minimized. Let $R(S, T, E_\textrm{usable})$ denote the overall risk, where $E_\textrm{usable}\subset E$ is the used set for steganographic communication, the optimization problem is then formulated as:
\begin{equation}
E_{\textrm{opt}}(S, T) = \underset{E_{\textrm{usable}}\subset E}{\textrm{arg min}}~~~R(S, T, E_\textrm{usable}).
\end{equation}

\section{Models and Solutions}
A path (if any) between two vertices $v_i$ and $v_j$ corresponds to a candidate communication chain over the network. A path with the minimum communication risk will be the preferred choice for the two vertices. We define the communication risk over a path as additive, meaning that, for arbitrary $L_k(v_i, v_j)\in P(v_i, v_j)$, $1\leq k\leq |P(v_i, v_j)|$, the communication risk is:
\begin{equation}
R(\{v_i\}, \{v_j\}, L_k(v_i, v_j)) = \sum_{e_i\in L_k(v_i, v_j)} w_i.
\end{equation}

The assigned weight $w_i$ for $e_i\in E$ reflects the possibility of being attacked, bandwidth cost and other quantifiable/abstract characteristics. A larger weight indicates that, the corresponding channel is less secure. The above additive assumption is reasonable since a minimized sum is intuitively corresponding to a lowest communication cost, even though the ground truth may be not exactly additive. In addition, the additive assumption could make the optimization target more clear and amenable to mathematical analysis or empirical study.

\subsection{Peer-to-Peer Steganographic Communication}
The simplest model is the peer-to-peer (P2P) steganographic communication problem, where $S = \{s_1\}$ and $T = T_1$ = $\{t_1\}$. In case that $w_i = w_j$ for all $1\leq i\neq j\leq m$, we have:
\begin{equation}
\begin{split}
E_{\textrm{opt}}(S,T) &= \underset{L_k(s_1,t_1)\in P(s_1,t_1)}{\textrm{arg min}}~R(\{s_1\}, \{t_1\}, L_k(s_1,t_1))\\
&= \underset{L_k(s_1,t_1)\in P(s_1,t_1)}{\textrm{arg min}}~\sum_{e_i\in L_k(s_1, t_1)} w_i \\
&= \underset{L_k(s_1,t_1)\in P(s_1,t_1)}{\textrm{arg min}}~|L_k(s_1, t_1)|,
\end{split}
\end{equation}
which requires us to find a path with the least number of edges. It is straightforward to use the basic graph-traversal technique called breadth-first search (BFS) \cite{algo:2009} to quickly identify the optimal path with a computational complexity of $O(|V|+|E|)$. BFS can be used for the design of steganography \cite{wu:isbast, wu:mtap}.

In case that there exists some $w_i\neq w_j$, it requires us to find the shortest path from $s_1$ to $t_1$. This can be addressed by Dijkstra's algorithm \cite{algo:2009}, which can be implemented with a computational complexity of $O(|E|+|V|\textrm{log}_2|V|)$ by using a priority queue data structure. It is noted that, the weights should be all non-negative here. Actually, one can also use other shortest path algorithms such as Floyd-Warshall algorithm \cite{algo:2009}, which allows the weights to be negative.

There may exist multiple paths with the minimum communication risk. In this case, it would be desirable to accept such a path that the number of vertices is minimum, which can be addressed by using a dynamic programming technique after calling the Dijkstra's algorithm. Let $f(v_i, v_j)$ be the number of vertices in such an optimal path (between $v_i$ and $v_j$) that it has the least number of vertices. Accordingly,
\begin{equation}
f(v_i, v_j) = \textrm{min}~\{f(v_i, v_k) + f(v_k, v_j) - 1~|~v_k\in V\}.
\end{equation}

Therefore, for the P2P steganographic communication network, we can determine the optimal communication chain out with a low computational complexity. Notice that, if there has no path between $v_i$ and $v_j$, no communication will be available between them. In steganography, though there has no explicit function between distortion and security, a lower embedding distortion often corresponds to a higher security level. If the distortion can be well modeled on a graph, we may use shortest path algorithms to achieve optimal or near-optimal embedding strategy such as Syndrome-Trellis Codes \cite{jessica:stcs}.

\subsection{Multi-Point Steganographic Communication}
It is more common that $|S| > 1$ and $|T| > 1$ in practice. Suppose that $S = \{s_1\}$ and $T = T_1 = \{t_1, t_2, ..., t_g\}$, an intuitive idea of finding the feasible communication strategy is to determine all shortest paths between $s_1$ and each one of $\{t_1, t_2, ..., t_g\}$. Let $E_\textrm{opt}(s_1, t_i)$ represent an optimal communication path between $s_1$ and $t_i$. If $E_\textrm{opt}(s_1, t_i)\cap$$E_\textrm{opt}(s_1, t_j) = \emptyset$ for all $i\neq j$, one may use $\cup_{i=1}^{g}E_\textrm{opt}(s_1, t_i)$ as the final strategy. It may be not optimal from risk optimization view since there may exist path redundancy. E.g., in Fig. 1 (a), $S = \{v_1\}$, $T = \{v_5, v_6\}$, $E_\textrm{opt}(v_1, v_5)$ $= \{(v_1, v_3), (v_3, v_4), (v_4, v_5)\}$ and $E_\textrm{opt}(v_1, v_6)$ $= \{(v_1, v_2), (v_2, v_4),$ $(v_4, v_6)\}$. If the encoder $v_1$ uses $E_\textrm{opt}(v_1, v_5)\cup E_\textrm{opt}(v_1, v_6)$ as the communication strategy, when he wishes to share the same message to both $v_5$ and $v_6$, he needs to send the message two times: one is conveyed by $v_1\rightarrow v_3\rightarrow  v_4\rightarrow v_5$ and the other is $v_1\rightarrow v_2\rightarrow  v_4\rightarrow v_6$. It is straightforward to prove that $w_1+w_2 = w_5+w_6$, meaning that, $\{(v_1, v_3), (v_3, v_4),$ $(v_4, v_6)\}$ and $\{(v_1, v_2), (v_2, v_4),$ $(v_4, v_5)\}$ are also the optimal paths for $v_6$ and $v_5$, respectively. It implies that, $v_1$ can send the message only once along $v_1\rightarrow v_2\rightarrow v_4$ and $v_4$ can distribute it to $v_5$ and $v_6$ (though $v_4$ cannot retrieve the embedded message). It is seen that, the overall risk is intuitively reduced as shown in Fig. 1 (b).

In case that $E_\textrm{opt}(s_1, t_i)\cap E_\textrm{opt}(s_1, t_j) \neq \emptyset$ for some $i\neq j$, $\cup_{i=1}^{g}E_\textrm{opt}(s_1, t_i)$ is also not optimal for the final strategy. That is, the edge(s) belonging to (at least) two paths will be used for transmission multiple times. It is reasonable since an encoder may send different messages along different paths, which may require the encoder to use the same edge multiple times. From the risk optimization view, the weights assigned to the edges are actually fusion results. It implies that, the impact of the used times of any edge has been (roughly) quantized by its weight. Thus, we may not care about the used times of edges directly. Accordingly, with the aforementioned additive assumption over a single path, the overall risk for multi-point communication is defined as:
\begin{equation}
\begin{split}
E_{\textrm{opt}}(S, T) &= \underset{E_{\textrm{usable}}\subset E}{\textrm{arg min}}~R(S, T, E_\textrm{usable})\\
&= \underset{E_{\textrm{usable}}\subset E}{\textrm{arg min}}~\sum_{e_i\in E_{\textrm{usable}}}w_i.
\end{split}
\end{equation}

Let $U_i = \{s_i\}\cup T_i$ for all $1\leq i\leq n_1$. We assume that, for all possible $U_i$, there should always exist such $j\neq i$ that $U_i\cap U_j\neq \emptyset$. The reason is that, for some $U_i$, if for all $j\neq i$, we have $U_i\cap U_j= \emptyset$. Then, we can split $(S, T)$ into two multi-point steganographic communication subtasks, i.e., $(S-\{s_i\}, T-T_i)$ and $(\{s_i\}, T_i)$.
{\color{red}{Therefore, it can be inferred that, under the above assumption, $U_1, U_2, ..., U_{n_1}$ can be merged into multiple disjoint vertex sets, denoted by $U_1^{(1)}, U_2^{(1)}, ..., U_{n_1^{(1)}}^{(1)}$ where $n_1^{(1)} < n_1$. By recursively applying the above assumption\footnote{{\color{red}{One can also assume that for all $2\leq i\leq n_1$, there should always exist such $j< i$ that $U_i\cap U_j\neq \emptyset$. Both ensure that, we do not need to split $(S, T)$ into multiple subtasks.}}} to $U_1^{(1)}, U_2^{(1)}, ..., U_{n_1^{(1)}}^{(1)}$ until there has only one set, we then write:}}
\begin{lemma}
$E_\textrm{opt}(S,T)$ will form a tree structure.
\end{lemma}
\begin{proof}
It is straightforward to prove that the resulting subgraph due to $E_\textrm{opt}(S,T)$ is connected. We have to prove that, there has no circle in the subgraph. Suppose that, there has a circle in the subgraph. One can always remove the edge with the largest weight from the circle such that the new subgraph will be still connected and the overall risk will not be worse. Therefore, $E_\textrm{opt}(S,T)$ forms a tree structure
\footnote{{\color{red}{Note that, when the proposed assumption is not required, $E_\textrm{opt}(S,T)$ will form a forest structure. A tree is a special case of a forest. Unless mentioned, in default, we think the proposed assumption holds.}}}.
\end{proof}
We begin with a simplified problem, where $V = S\cup T$. The optimal solution is equivalent to finding minimum spanning tree (MST), which can be solved by Kruskal's algorithm or Prim's algorithm \cite{algo:2009}. It is mentioned that, the greedy strategy for computing MST has been exploited in the design of steganography such as optimal parity assignment \cite{jessica:opa}.
\begin{figure}[!t]
\centering
\includegraphics[width=3.5in]{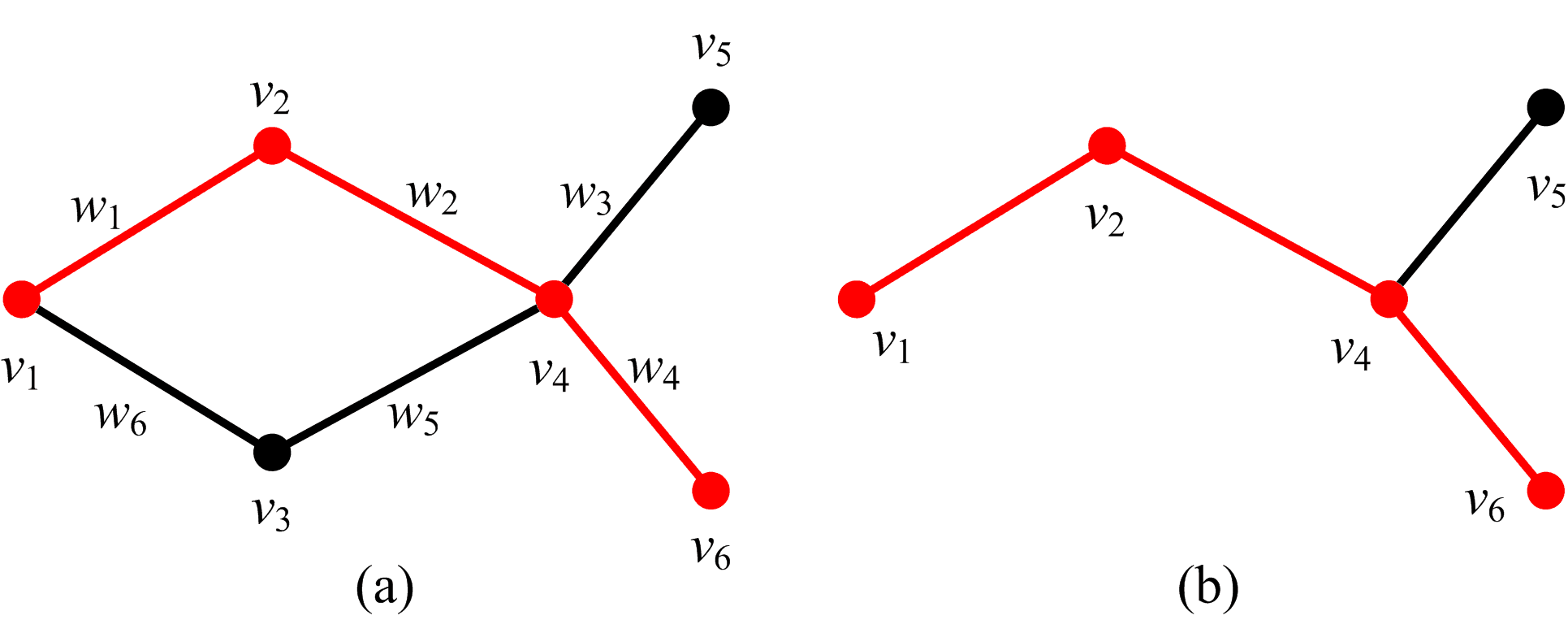}
\caption{An example of multi-point steganographic communication.}
\end{figure}

Actually, the original problem to be optimized is essentially corresponding to such a graph problem: given an undirected graph $G(V, E)$ without negative edge-weights and a subset of vertices
$S\cup T$ called terminals, the goal is to find such a tree of minimum weight-sum that it contains all the terminals (but may include additional vertices). This problem is typically called as Steiner Tree Problem (STP), which is NP-hard \cite{chlebik:stp}. The STP is seen as a generalization of two other combinatorial optimization problems mentioned above, i.e., the shortest path problem and the MST problem. If the STP deals with only two terminals, it reduces to finding a shortest path. If all vertices are terminals, it is equivalent to the MST problem. Therefore, one may employ the existing approximation algorithms \cite{chlebik:stp, berman:stp, byrka:lp4stp} designed for the STP to find the suitable strategy for multi-point steganographic communication.

\subsection{Probabilistic Graph Perspective}
We have used a weight to expose the communication risk of a channel. It enables us to deal with different steganographic scenarios by using the shortest path algorithm, MST algorithm, or approximation algorithms designed for the STP. Considering real-world applications, it would be more intuitive to use statistical information to characterize the reliability of a channel. The reason is, in practice, the raw data we may obtain is more like a series of successful or failure records, which may not only rely on some specified steganographic algorithm.

By analyzing statistical characteristics of a series of records, one may roughly estimate the prior probability for the communication optimization. There are at least two advantages. First, prior probability does not only rely on some specified steganographic algorithm. Second, probabilistic analysis will be more robust. Suppose that, the steganographic social network are modified with probabilistic edges, meaning that, the edges are associated with a probability, rather than a weight, denoting the successful probability of steganographic communication, which is derived from the real-world records (even though we have no real data at the moment). Accordingly, the reliability of a path $L_k(v_i, v_j)\in P(v_i, v_j)$ between $v_i$ and $v_j$ can be assumed as the product of probability values along the path:
\begin{equation}
J(\{v_i\}, \{v_j\}, L_k(v_i, v_j)) = \prod_{e_i\in L_k(v_i, v_j)} p_i,
\end{equation}
where $p_i$ denotes the assigned probability of $e_i$. The multiplicative assumption is reasonable in practice since two edges corresponding to two different channels could be approximately independent. Similarly, for $G(V, E)$, we further specify:
\begin{equation}
\begin{split}
E_{\textrm{opt}}(S, T) &= \underset{E_{\textrm{usable}}\subset E}{\textrm{arg max}}~J(S, T, E_\textrm{usable})\\
&= \underset{E_{\textrm{usable}}\subset E}{\textrm{arg max}}~\prod_{e_i\in E_{\textrm{usable}}}p_i,
\end{split}
\end{equation}
which requires us to find a set of edges such that the overall reliability (multiplicative) is maximum. We will prove that:
\begin{lemma}
$E_\textrm{opt}(S,T)$ will also form a tree structure, and the solution is equivalent to solving STP on a weighted graph.
\end{lemma}
\begin{proof}
We write:
\begin{equation}
\begin{split}
E_{\textrm{opt}}(S, T) &= \underset{E_{\textrm{usable}}\subset E}{\textrm{arg max}}~\prod_{e_i\in E_{\textrm{usable}}}p_i\\
&= \underset{E_{\textrm{usable}}\subset E}{\textrm{arg min}}~-\textrm{log}_2\prod_{e_i\in E_{\textrm{usable}}}p_i\\
&= \underset{E_{\textrm{usable}}\subset E}{\textrm{arg min}}~\sum_{e_i\in E_{\textrm{usable}}}\textrm{log}_2\frac{1}{p_i}.\\
\end{split}
\end{equation}

By assigning $w_i = \textrm{log}_2\frac{1}{p_i}$, we can translate the probabilistic graph into a weighted graph so that our optimization task is exactly equivalent to solving STP in the weighted graph \footnote{{\color{red}{Here, the aforementioned assumption for Lemma III.1 should hold. Otherwise, it would form a forest, which is not corresponding to the STP.}}}.
\end{proof}

A multiplicative probabilistic graph is therefore corresponding to an additive weighted graph. The optimization solutions are similar to each other.

\section{Detection Analysis}
An attacker may serve as a network manager (or monitor), who has the surveillance right. He hopes to detect whether the steganographic communication exist or not over the network. A straightforward idea is to use the existing advanced steganalysis algorithms \cite{jessica:srm, xu:sdcnn, xu:ensemble, xu:cnnJUNIWARD, ye:cnn}, which mainly rely on sophisticated manual feature design, conventional classifiers and deep neural networks. Accordingly, the detection is equivalent to designing reliable steganalysis systems. In other words, to resist against conventional steganalysis algorithms over a social network, we should use statistical undetectable steganographic systems or other efficient strategies, e.g., batch steganography \cite{ker:batch, ker:move}. We will not study it in this paper since there are a lot of related works that have been reported.

From the viewpoint of social network, we here consider two attacks, i.e, structural attack and statistical attack. The former analyzes the structural characteristics of a graph without considering external effects, which allows us to easily analyze problems in an offline manner. The latter will take into account the communication activities in the network. It would be more useful than the structural analysis.

\subsection{Structural Attack}
Structural attack relies on the network structure. It is probably not suitable for detecting the data encoder(s) and decoder(s) ``\emph{directly}'' sometimes, but to identify the channels denoted by edges. We define $\delta(e)$ as the number of such pair $(v_i, v_j)$ that $e\in E$ belongs to at least one shortest path between $v_i$ and $v_j$, i.e.,
\begin{equation}
\delta(e) = |\{(v_i, v_j) | 1\leq i < j\leq n, e\in P_\textrm{s}(v_i, v_j)\} |,
\end{equation}
where $P_\textrm{s}(v_i, v_j)$ is the set including all shortest paths between $v_i$ and $v_j$. We define the \emph{path-support rate (PSR)} of $e$ as:
\begin{equation}
r(e) = \frac{\delta(e)}{\sum_{e'\in E}\delta(e')}.
\end{equation}

A higher PSR indicates that it is intuitively more likely to be exploited for steganographic communication. For detection, more attention should be paid to edges with a high PSR. We conduct experiments with \emph{virtual data} to qualitatively describe the detection performance. Suppose that, an attacker selects a subset of $E$, denoted by $E_\textrm{attack}$, according to PSR. Each edge in $E_\textrm{attack}$ corresponds to a communication channel. Traditional steganalysis algorithms will be applied to the detection of these channels. We define the \emph{edge-selection rate (ESR)} as the size-ratio between $E_\textrm{attack}$ and $E$, i.e., $|E_\textrm{attack}|/|E|$. To determine $E_\textrm{attack}$, the PSR values for all edges are sorted in an increasing order, and the top-$|E_\textrm{attack}|$ edges constitute $E_\textrm{attack}$. We simulate P2P steganographic communication on a network by \emph{randomly} choosing two different vertices and determining a shortest path between them (by BFS). In a path, if there is at least one edge belonging to $E_\textrm{attack}$, we say the communication path will be at the risk of being detected by such as steganalysis algorithms.

For a random network, we randomly generate a number of vertex-pairs and compute their shortest paths. We define the \emph{path-hit rate (PHR)} as the ratio between the number of such shortest paths that they are at risk of being detected and the number of tested vertex-pairs. We use PHR as the detection metric. Fig. 2 shows the results for different random networks. Due to the lack of real data, we cannot guarantee that Fig. 2 characterizes the real-world scenarios accurately. However, we can still draw the potential rule of random data, which may be helpful for future study. That is, for a fixed number of vertices, a larger number of edges results in a worse detection performance. When the number of edges approaches to the upper-bound, the detection curve will be rather close to $y = x$, which means \emph{random guessing}. It is seen that, with low ESR values (which may be often used in applications), the detection performance is satisfied (acceptable) to the steganographer(s). However, we admit that it may be not significant since the real steganographic network may not follow random structure.

The structural attack is largely affected and reflected by the topological structure of the network, which involves degrees, connectedness, path planning and relationship of vertices. Any efficient algorithms focusing on identifying critical or unusual edges and vertices are likely suited to steganographic network.

\subsection{Statistical Attack}
The structural attack may result in a high missing detection rate since it does not take into account the real data. For data-driven statistical attack, it allows an attacker to mine unusual information from real records, which is more robust. We here present limited discussion for the statistical attack.

\emph{1) Data-driven Structural Analysis:} This is similar to the above structural attack. A straightforward idea is to construct a subnetwork according to the real-data flow. Then, structural analysis may be used to identify the data encoder(s) and data decoder(s). Notice that, the subnetwork may contain multiple connected components. Comparing with the above structural attack, data-driven structural analysis takes into account the real communication activities and features. It can reduce the impact of vertices that never communicate with others.

\emph{2) Data-driven Statistical Analysis:} The explicit statistical characteristics (or distribution) of steganographic communication is unknown to us. We expect to exploit the communication activities to identify the anomalies. We divide the data-driven statistical analysis into three categories roughly, i.e., behavior analysis, content analysis, unsupervised learning analysis.

Content-based detection requires us to construct the differentiable statistical information from any form of the concerned objects, for which it is common to use efficient classifiers such as neural networks and SVM or other boosting strategies, e.g., ensemble, to distinguish ``normal'' and ``abnormal'' from the extracted feature vectors or feature maps. The traditional steganalysis algorithms can be therefore easily classified as a core approach to content analysis. Unlike traditional steganalysis methods, general content analysis on a network may involve  a set of vertices or edges.

The behavior analysis is not intuitive in the steganographic network. A common framework in steganography is to embed secret data into a specified cover. The stego object is then sent to a receiver, for whom the embedded data can be retrieved. We call it as cover-based steganography. Due to the widely use of social medias, one may also apply behavior steganography \cite{zhang:behavior} (a kind of coverless steganography) for communication. It is necessary to identify those suspicious vertices from unusual behavior observations. The social behavior can be quantized as analyzable parameters associated to network vertices or edges. Also, they may be used to construct a new behavior network for analysis. Feature design would be critical for both cases.
\begin{figure}[!t]
\centering
\includegraphics[width=3.2in]{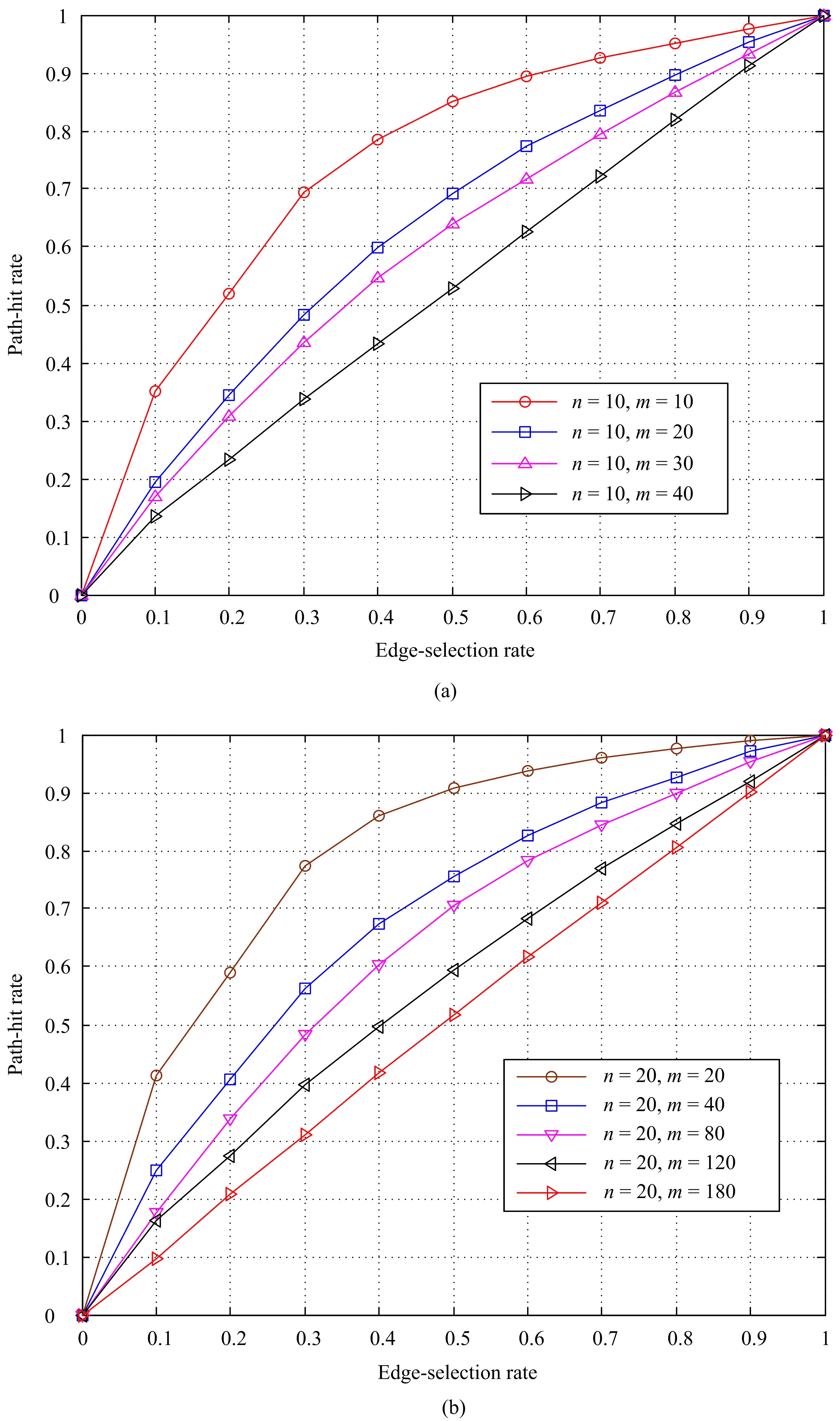}
\caption{The detection performance with structural analysis for P2P steganographic communication on \emph{randomly generated networks}.}
\end{figure}

The content analysis and the behavior analysis often require supervised learning. On the one hand, it is not easy to obtain the real sample labels. On the other hand, the time-varying network may lead to a low detection accuracy. Though the raw structural analysis corresponds to unsupervised learning, it focuses on the topological characteristics of the network and is independent of real data. There are at least two types for unsupervised learning analysis, i.e., clustering-based and model-based. According to the extracted features from the raw data, clustering-based methods treat the classification task as the clustering problem in data mining. While, the model-based methods may assume that ordinary activities meet some rules and any abnormal behavior will break them down.

The existing anomaly detection approaches and perspectives for social networks probably also work for a steganographic network though they were designed to other social scenarios originally. However, more attention should be paid to the characteristics of the steganographic behavior itself.

\section{Conclusion and Future Works}
A preferred starting point of scientific research is to make a model for the problem \cite{ker:move}. In this paper, we model steganography based communication on a network. The optimization problem is to minimize the overall secret communication risk, which is considered as additive over the network. We have analyzed different scenarios, and provided the optimized solutions. With real-world data, one may estimate the prior information for an insecure channel, which reveals the reliability of channels. By translating the probabilistic network into a weighted network, one could find the optimal (P2P) or near-optimal (multi-point) communication strategy. We also present some discussion from detection view. The future considerations and works include:

\begin{itemize}
\item A single weight or probability would make the problem amenable to mathematical analysis. It may not characterize the steganographic channel very well. Feature vectors may be more desirable.
\item A basic theoretical model could be more robust and less ambiguous. However, the real-world is more messy and complex. It would be important to take into account the attackers and other constraints in the social network.
\item A challenging problem is to build an adequate and real-world benchmark for experiments. A compromised way focusing on synthetic data may be a good choice.
\item Moving steganography from laboratory into real-world is a system engineering. Not only the underlying technologies should be analyzed, but also such as social behaviors should be studied, which may require knowledge about management, psychology, network analysis etc.
\end{itemize}




%

\end{document}